\documentclass[sigconf, anonymous=false, review=false]{acmart}

 \usepackage{amsmath}
 \usepackage{graphicx}
\usepackage{url}            
\usepackage{booktabs}       
\usepackage{nicefrac}       
\usepackage[mathscr]{euscript}
\usepackage{microtype}      
\usepackage{mathrsfs}
\usepackage{textcomp}
\usepackage{subfig}
\usepackage{epsfig}
\usepackage{amssymb}
\usepackage{amsthm}
\usepackage{algorithmic}
\usepackage{algorithm}
\usepackage{multirow}
\usepackage{mathtools}
\usepackage{bbm}
\usepackage[utf8]{inputenc}
\usepackage[english]{babel}

\usepackage[mathscr]{euscript}
 \let\mathscr\relax
\usepackage[scr]{rsfso}
 
\newtheorem{theorem}{Theorem}[section]

\newtheorem{lemma}[theorem]{Lemma}

\usepackage{algorithm,algorithmic}

\begin{document}
\newcommand{\zf}[1]{\textcolor{red}{[\small zf: ~#1~]}}
\title[]{ A Unified Approach to Scalable Spectral Sparsification of Directed Graphs}

\author{Ying Zhang}
\affiliation{
\institution{Stevens Institute of Technology}
\city{Hoboken}
\state{New Jersey}
\postcode{07030}
}
\email{yzhan232@stevens.edu}

\author{Zhiqiang Zhao}
\affiliation{
\institution{Michigan Technological University}
\city{Houghton}
\state{Michigan}
\postcode{49931}
}
\email{qzzhao@mtu.edu}

\author{Zhuo Feng}
\affiliation{
\institution{Stevens Institute of Technology}
\city{Hoboken}
\state{New Jersey}
\postcode{07030}
}
\email{zhuo.feng@stevens.edu}
\begin{abstract}
 Recent spectral graph sparsification research allows constructing nearly-linear-sized subgraphs that can well preserve the spectral (structural) properties of the original graph, such as the first few eigenvalues and eigenvectors of the graph Laplacian, leading to the development of a variety of nearly-linear time numerical and graph algorithms. However,  there is not a unified approach that allows for truly-scalable spectral sparsification of both directed and undirected graphs. In this work, we prove the existence of linear-sized spectral sparsifiers for general directed graphs and introduce a practically-efficient and unified spectral graph sparsification approach that allows sparsifying real-world, large-scale directed and undirected graphs with guaranteed preservation of the original graph spectra. By exploiting a highly-scalable (nearly-linear complexity) spectral matrix perturbation analysis framework for constructing nearly-linear sized (directed) subgraphs,  it enables us to well preserve the key eigenvalues and eigenvectors of the original (directed) graph Laplacians. The proposed method has been validated using various kinds of directed graphs obtained from public domain sparse matrix collections, showing promising results for solving directed graph Laplacians, spectral embedding, and partitioning of general directed graphs, as well as approximately computing (personalized) PageRank vectors.
\end{abstract}

\maketitle
\section{Introduction}\label{sect:introduction}
Many  research problems for {simplifying large graphs leveraging spectral graph theory}  have been  extensively studied by mathematics and theoretical computer science (TCS) researchers in the past decade \cite{ batson2012twice, spielman2011spectral, kolev2015note,peng2015partitioning, cohen2017almost,Lee:2017,cohen:focs18}.  Recent \emph{{spectral graph sparsification }} research allows   constructing nearly-linear-sized \footnote{The number of edges is close to the number of nodes in the sparsifier.} subgraphs that can well preserve the spectral (structural) properties of the original graph, such as the  the first few eigenvalues and eigenvectors of the graph Laplacian. The related results can potentially lead to the development of a variety of {\emph{nearly-linear time}}  numerical and graph algorithms for solving large sparse  matrices and partial differential equations (PDEs), graph-based semi-supervised learning (SSL), computing the  stationary distributions of Markov chains and personalized PageRank vectors, spectral graph partitioning and data clustering, max flow and multi-commodity flow of undirected graphs, nearly-linear time circuit simulation and verification algorithms, etc \cite{miller:2010focs, spielman2011spectral, christiano2011flow, spielman2014sdd,kelner2014almost,cohen2017almost,cohen:focs18,zhuo:dac16,zhiqiang:dac17,zhiqiang:iccad17, zhuo:dac18}.

However,  {there is not a unified approach that allows for truly-scalable spectral sparsification of both  directed and undirected graphs}. For example, the state-of-the-art  sampling-based methods  for spectral  sparsification   are only applicable to undirected graphs \cite{spielman2011graph,miller:2010focs,spielman2014sdd}; the latest algorithmic breakthrough in spectral sparsification of directed graphs \cite{cohen2017almost,cohen:focs18} can only handle strongly-connected directed graphs \footnote{A strongly connected directed graph is a directed graph in which any node can be reached from any other node along with direction.}, which inevitably limits its applications when confronting  real-world graphs, since many directed graphs may not be strongly connected, such as the graphs used in chip design automation (e.g. timing analysis) tasks as well as the graphs used in machine learning and data mining  tasks. Consequently, there is still a pressing need for the development of highly-robust (theoretically-rigorous) and truly-scalable (nearly-linear complexity) algorithms for reducing real-world large-scale (undirected and directed) graphs while preserving key graph spectral (structural) properties.

This paper  proves the existence of linear-sized spectral sparsifiers for general directed graphs, and introduces  a {{practically-efficient and unified spectral sparsification approach}} that allows simplifying real-world, large-scale  directed and undirected graphs with guaranteed preservation of the original graph spectra. More specifically, we  exploit a highly-scalable (nearly-linear complexity) spectral matrix perturbation analysis framework for constructing ultra-sparse (directed) subgraphs that can well preserve the key eigenvalues and eigenvectors of the original graph Laplacians. Unlike the prior state-of-the-art methods  that are only suitable for handling specific types of graphs (e.g. undirected or strongly-connected directed graphs \cite{spielman2011graph,cohen2017almost}), the proposed approach is more general and thus will allow for truly-scalable spectral sparsification of a much wider range of real-world complex graphs that may involve billions of elements. The spectrally-sparsified directed graphs constructed by the proposed approach will potentially lead to the development of much faster numerical and graph-related algorithms. For example, spectrally-sparsified social (data) networks allow for more efficient  modeling and analysis of large social (data) networks; spectrally-sparsified neural networks  allow for  more scalable model training and processing in emerging machine learning tasks; spectrally-sparsified web-graphs allow for much faster computations of personalized PageRank vectors; spectrally-sparsified integrated circuit networks will lead to more efficient partitioning, modeling, simulation, optimization and verification of large chip designs, etc.

The rest of this paper is organized as follows.  Section
\ref{background_sec} provides a brief introduction to   graph Laplacians and spectral sparsification of directed graphs.  In Section \ref{sec:maintech}, a scalable and unified spectral sparsification framework for general graphs is described in detail. Section \ref{sec:practical}   describes a practically-efficient spectral sparsification approach, while Section \ref{sec:application}  introduces potential applications of the proposed graph sparsification framework. Section
\ref{sec:result}  demonstrates extensive  experimental results for a variety of real-world, large-scale directed graphs, which is followed by the conclusion of this work in Section \ref{conclusion}.
\section{Preliminaries}\label{background_sec}

\subsection{Laplacians for  (un)directed graphs}
 Consider a directed graph $G=(V,E_G,w_G)$ with $V$ denoting the set of vertices, $E_G$ representing the set of directed edges, and $w_G$ denoting the associated edge weights. In the following, we denote the diagonal matrix by $\mathbf{D_G}$ with ${D_G}(i,i)$ being equal to the (weighted) outdegree of node $i$, as well as the adjacency matrix of $G$ by $\mathbf{A_G}$:
\begin{equation}\label{di_laplacian}
A_G(i,j)=\begin{cases}
w_{ij} & \text{ if } (i,j)\in E_G \\
0 & \text{otherwise }.
\end{cases}
\end{equation}
Then  the directed Laplacian  matrix can be constructed as follows \cite{cohen2017almost}:
\begin{equation}\label{formula_laplacian}
\mathbf{L_G=D_G-A_G^\top}.
\end{equation}
Let $n=|V|$, $m=|E_G|$,  and undirected graphs  can be converted into  equivalent directed graphs by replacing each undirected edge with two opposite directed edges. While for most directed graphs $\mathbf{L_G}$ may not be a symmetric matrix. 

It can be shown that  any directed (undirected) graph Laplacian constructed using (\ref{formula_laplacian}) will satisfy the following properties: \textbf{1)} Each column (and row) sum is equal to zero; \textbf{2)} All off-diagonal elements are non-positive; \textbf{3)} The Laplacian matrix is asymmetric (symmetric) and indefinite (positive semidefinite). 

\subsection{Spectral   sparsification of undirected graphs}
Graph sparsification aims to find a subgraph (sparsifier) $S=(V,E_S,w_S)$ that has the same set of vertices but  much fewer edges than the original graph $G$. There are two types of  sparsification methods: the  cut sparsification methods  preserve  cuts in the original graph through random sampling of edges \cite{benczur1996approximating}, whereas spectral sparsification methods preserve the graph spectral (structural) properties, such as distances between vertices, effective resistances, cuts in the graph, as well as the  stationary distributions of Markov chains \cite{cohen2017almost,cohen:focs18,spielman2011spectral}.  Therefore, spectral graph sparsification is a much stronger notion than cut sparsification. 

{For {undirected graphs}}, spectral sparsification aims to find an ultra-sparse subgraph proxy that is {spectrally-similar} to the original one.  $G$ and $S$ are  said to be {$\sigma$-spectrally similar} if the following condition holds for all real vectors $\mathbf{x} \in \mathbb{R}^V$: $\frac{\mathbf{x}^\top{\mathbf{L_S}}\mathbf{x}}{\sigma}\le \mathbf{x}^\top{\mathbf{L_G}}\mathbf{x} \le \sigma \mathbf{x}^\top{\mathbf{L_{S}}}\mathbf{x}$, where $\mathbf{L_{G}}$ and $\mathbf{L_S}$ denote the symmetric diagonally dominant (SDD) Laplacian matrices  of graphs $G$ and $S$, respectively. By defining the relative condition number to be $\kappa({\mathbf{L_G}},{\mathbf{L_S}})=\lambda_{\max}/\lambda_{\min}$, where $\lambda_{\max}$ ($\lambda_{\min}$) denotes the largest (smallest) eigenvalues of $\mathbf{L_S^{+} L_G}$, and $\mathbf{L_S^+}$ denotes the Moore-Penrose pseudoinverse of $\mathbf{L_S}$,  it can be further shown that $\kappa(\mathbf{L_G},\mathbf{L_S})\le\sigma^2$, implying that a smaller relative condition number or $\sigma^2$ corresponds to a higher (better) spectral similarity between two graphs.  
\subsection{Spectral   sparsification of directed graphs} 
A significant progress has been made for spectral analysis of {directed graphs }in \cite{chung2005laplacians}, which for the first time has proved the Cheeger inequality for directed graphs and  shown the connection between directed graph partitioning   and the smallest (nontrivial) eigenvalue of directed Laplacian. More specifically, the transition probability matrix and the stationary distributions of Markov chains have been exploited for constructing the undirected Laplacians for strongly-connected directed graphs. 
The latest algorithmic breakthrough in spectral sparsification   for strongly-connected directed  graphs has been introduced based on the results in \cite{chung2005laplacians}, which proposes to first convert  strongly connected graphs into  Eulerian graphs via Eulerian scaling, and subsequently sparsify the undirected graphs obtained via  Laplacian symmetrization \cite{chung2005laplacians} by leveraging existing spectral graph theory for undirected graphs \cite{cohen2017almost}. It has been shown that such an approach for directed graphs can potentially lead to the development of almost-linear-time algorithms for  solving asymmetric linear systems, computing the stationary distribution of a Markov chain, computing expected commute times in a directed graph, etc \cite{cohen2017almost}.

 {For { {directed graphs}}}, the subgraph $S$ can be considered spectrally similar to the original  graph $G$ if the    condition number or the ratio between the largest and smallest singular values of $\mathbf{L_S^{+} L_G}$ is close to $1$ \cite{cohen2017almost,cohen:focs18}.  Since the singular values of $\mathbf{L_S^{+} L_G}$ correspond to  the square roots of eigenvalues of  $\mathbf{(L_S^{+}L_G )^\top(L_S^{+}L_G )}$,   {spectral sparsification of directed graphs} is equivalent to finding an ultra-sparse subgraph $S$ such that the  condition number of $\mathbf{(L_S^{+}L_G )^\top(L_S^{+}L_G )}$ is small enough.

\section{A Unified  Sparsification Framework} \label{sec:maintech}
\subsection{Overview of our approach} 
We introduce {a unified spectral graph sparsification framework} that allows handling  both directed and undirected graphs in nearly-linear time. The core idea of our approach is to leverage a novel {spectrum-preserving Laplacian symmetrization} procedure to convert   directed graphs  into  undirected ones (as shown in Figure \ref{fig:symmetrization}).  Then existing spectral sparsification methods for undirected graphs \cite{spielman2011spectral,batson2012twice, zhuo:dac16,zhuo:dac18,Lee:2017} can be exploited for directed graph spectral sparsification tasks.  

Our approach for symmetrizing directed graph Laplacians is motivated by the following fact:  the eigenvalues of $\mathbf{(L_S^{+}L_G )^\top(L_S^{+}L_G )}$ will always correspond to the eigenvalues of $\mathbf{\left(L_S L^\top_S\right)^+ L_G L_G^\top}$ under the condition that $\mathbf{L_G}$ and $\mathbf{L_S}$ are diagonalizable. It can be shown that $\mathbf{L_G L^\top_G}$ and $\mathbf{L_S L^\top_S}$ can be considered as  special  graph Laplacian matrices corresponding to   undirected graphs that may contain negative edge weights. Consequently, as long as  a directed subgraph $S$ can be found  such that  the  undirected graphs corresponding to $\mathbf{L_S L^\top_S}$ and  $\mathbf{L_G L_G^\top}$ are spectrally similar to each other, the subgraph $S$  can be considered spectrally similar to the original directed graph $G$.  Unlike the recent theoretical breakthrough in directed graph sparsification \cite{cohen2017almost,cohen:focs18}, our approach    does not require the underlying directed graphs to be strongly connected, and thus can be applied to a much wider range of large-scale real-world problems, such as the neural networks adopted in many machine learning and data mining applications \cite{goodfellow2016deep,Ying:kdd18}, the directed graphs (e.g. timing graphs) used in various circuit analysis and optimization tasks \cite{micheli1994synthesis},  etc.
 \begin{figure}
\centering \epsfig{file=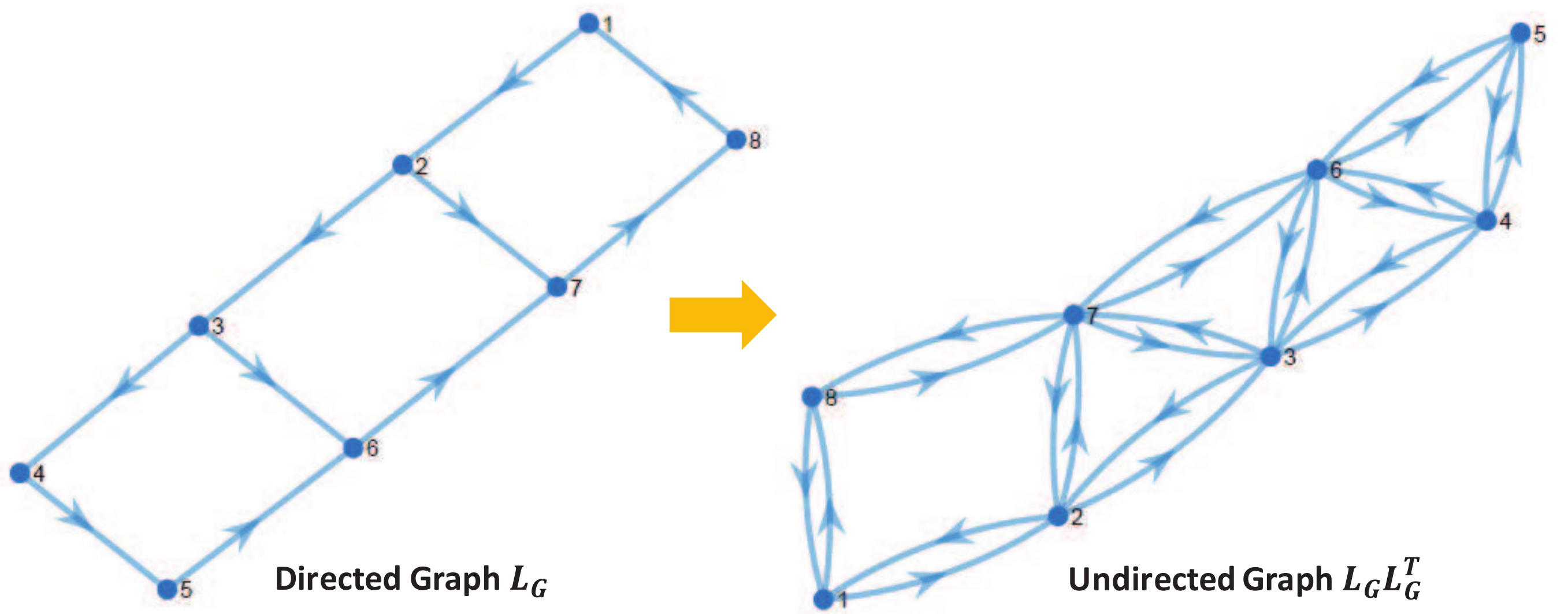, scale=0.32} \caption{Converting a directed graph $G$  into an undirected graph $G_u$  via Laplacian symmetrization. \protect\label{fig:symmetrization}}
\end{figure}

In the following,   assume that $G=(V,E_G, {w_G})$ is a weighted directed graph, whereas $S=(V,E_S, {w}_S)$ is its initial spectral sparsifier (subgraph), such as a spanning  subgraph. Define $\mathbf{L_{G_u}=L_G L_G^\top}$ and $\mathbf{L_{S_u}= L_S L_S^\top}$ to be the  undirected graph Laplacians obtained through the proposed  symmetrization procedure for   $G$ and $S$.

\subsection{Spectrum-preserving  symmetrization}
Performing singular value decomposition (SVD) on $\mathbf{L_G}$ leads to $\mathbf{L_G}=\sum _i \sigma _i \mathbf{\zeta_i}{\mathbf{ \eta^\top_i}}$, where $\mathbf{\zeta_i}$ and $\mathbf{\eta_i}$ are the left and right eigenvectors of $\mathbf{L_G}$, respectively. \footnote{The pseudoinverse of $\mathbf{L_G}$  is $\mathbf{{L_G}^+=\sum _i \frac{1}{\sigma _i}\eta_i \zeta^\top_i}$} It should be noted that $\mathbf{\zeta_i}$ and $\mathbf{\eta_i}$ with $i=1,...,n$ span the eigenspace of $\mathbf{L_G}{\mathbf{L^\top_G}}$ and ${\mathbf{L^\top_G}}\mathbf{L_G}$, respectively. Since the  eigenspace related to outgoing edges of directed graphs needs to be preserved,  we will only focus on the  Laplacian symmetrization  matrix $\mathbf{L_{G_u}}={\mathbf{L_G}}\mathbf{L_G}^\top$ that is also a  Symmetric Positive Semi-definite (SPS) matrix.

\begin{figure}
\centering \epsfig{file=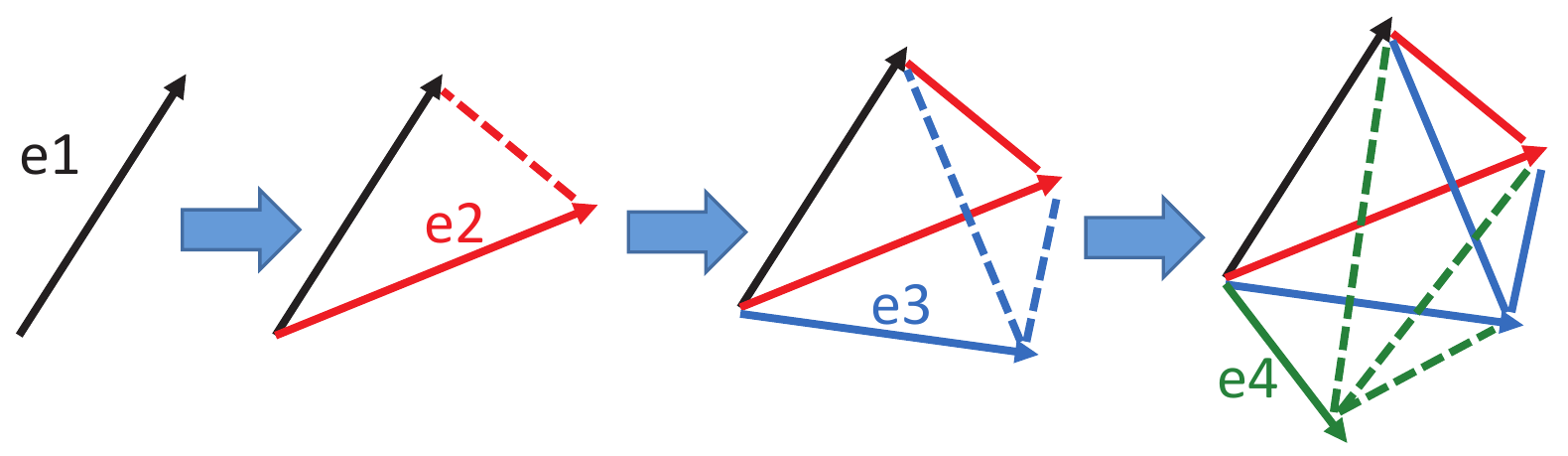, scale=0.51} \caption{Edge coupling during directed Laplacian symmetrization. \protect\label{fig:symetrization}}
\end{figure}

\begin{theorem}\label{thm:directed symmetrization}
For any directed Laplacian $\mathbf{L_G}$, its  undirected graph Laplacian $\mathbf{L_{G_u}}$ after symmetrization will have the all-one vector as its null space and correspond to an undirected graph that may include negative edge weights. 
\end{theorem}

\begin{proof} 
Each element $(i,j)$ in $\mathbf{L_{G_u}}$ can be written as follows:
\begin{equation}\label{di_laplacian_entry}
{L_{G_u}}_{ij}=\begin{cases}
{D_G}^2_{ii}+\sum_{k}{A_G}_{ki}^2 & \text{$i=j$}\\
\sum_{k}\left(-{A_G}_{ki}{A_G}_{kj}+{A_G}_{ki}{D_G}_{kj}+{D_G}_{ki}{A_G}_{kj}\right) & \text{$i\neq j$}.
\end{cases}
\end{equation}

It can be shown  that the following is always true:
\begin{equation}
\begin{split}
{L_{G_u}}_{ii}+\sum_{j, j \neq i}{L_{G_u}}_{ij}  &=\sum _k {L_G}_{ik}{L_G}_{ik}+\sum_{j, j\neq i}\sum _k {L_G}_{jk}{L_G}_{ik}\\
& =\sum _k {L_G}_{ik}\left({L_G}_{ik}+\sum_{j, j\neq i} {L_G}_{jk}\right)=0,
\end{split}
\end{equation}
which indicates the all-one vector is the null space of $\mathbf{L_{G_u}}$. 
For {directed graphs},  it can be shown that if a node has more than one outgoing edge, in the worst case the neighboring nodes pointed by such outgoing edges will form a clique possibly with negative edge weights in the corresponding undirected graph after symmetrization.  

As an example shown in Figure \ref{fig:symetrization}, when edge $e2$ is added into the initial graph $G$ that includes a single edge $e1$,  an extra edge (shown in red dashed line) coupling with $e1$ will be created in the resultant undirected graph $G_u$; similarly, when an edge $e3$ is further added,  two extra edges coupling with $e1$ and $e2$ will be created in $G_u$. When the last edge $e4$ is added, It forms a clique. 

It can be shown that ${G_u}$ will contain negative edge weights under the following condition:
\begin{equation}
\sum_{k}({A_G}_{ki}{D_G}_{kj}+{D_G}_{ik}{A_G}_{jk})>\sum_{k}{A_G}_{ki}{A_G}_{kj}.
\end{equation}
In some cases, there may exist no  clique even though all outgoing edges according to one node are added into subgraph only because the weights of these edges become zeros satisfying:  $\sum_{k}({A_G}_{ki}{D_G}_{kj}+{D_G}_{ik}{A_G}_{jk})=\sum_{k}{A_G}_{ki}{A_G}_{kj}$.
\end{proof}

\subsection{Existence of linear-sized spectral sparsifier}
It has been shown  that every undirected graph with positive edge weights has a Twice-Ramanujan spectral sparsifier with positive edge weights to spectrally-approximate the original graph \cite{batson2012twice,Lee:2017}. In this work,  we   extend the above theory  to deal with the undirected graphs obtained through the proposed Laplacian symmetrization procedure that may introduce negative weights.

\begin{theorem} \label{twice}
For a given directed graph $ G$ and its undirected graph ${G_u=(V,E_{G_u},w_{G_u})}$ obtained via Laplacian symmetrization, there exists a $(1+\epsilon)$-spectral sparsifier $S$ with $O(n/\epsilon ^2)$ edges such that its undirected graph $S_u=(V,E_{S_u}, {w_{S_u}})$ after symmetrization  satisfies the following condition for any $\mathbf{{x} \in \mathbb{R}^n}$:
\begin{equation}\label{twice_sym}
(1-\epsilon)\mathbf{x^\top L_{G_u}x \le x^\top L_{S_u}x  \le} (1+\epsilon) \mathbf{x^\top L_{G_u} x}. 
 \end{equation}

\end{theorem}

We will need the following   the lemma [\ref{theorem:dual_set}] to prove our theorem. 
\begin{lemma}\label{theorem:dual_set}
Let $\epsilon>0$, and  $\mathbf{u_1,u_2,...,u_{m}}$ denote a set of vectors in $\mathbb{R}^{n}$ that allow expressing the identity decomposition as: 
\begin{equation}
\mathbf{\sum_{1\leq i\leq m}u_i u_i^\top=I_{n\times n}},
\end{equation}
\end{lemma}
where $ \mathbf{I_{n\times n} \in \mathbb{R}^{{n\times n}}}$ denotes an identity matrix. Then there exists a $O(m/\epsilon^{(1)})$-time algorithm \cite{Lee:2017} that can find non-negative coefficients $\{t_i\}^m_{i=1}$  such that at most $|\{t_i|t_i> 0\}|=O(n/\epsilon ^2)$ and  for any $\mathbf{x} \in \mathbb{R}^n$: 
\begin{equation}\label{twice_sparsifier}
(1-\epsilon)\mathbf{ x^\top I_{n\times n} x\leq \sum_{i} t_i x^\top u_i u_i^\top x \leq}(1+\epsilon) \mathbf{ x^\top I_{n\times n} x}.
\end{equation}

\begin{proof}
Any directed graph Laplacian can  also be written  as:
\begin{equation}\label{formula_laplacian2}
\mathbf{L_G=B^\top W C}, 
\end{equation}
where $\mathbf{B}_{m\times n}$  and $\mathbf{C}_{m\times n}$ are the  edge-vertex incidence matrix and  the injection matrix defined below:
\begin{equation}
B(i,v)=\begin{cases}
1 & \text{ if } v \text{ is  i-th edge's    head}\\
-1 & \text{ if } v \text{ is   i-th edge's tail}, 
\end{cases}
\end{equation}
\begin{equation}
C(i,v)=\begin{cases}
1 & \text{ if } v \text{ is  i-th edge's    head}\\
0 & \text{ if } v \text{ is  i-th edge's tail}, 
\end{cases}
\end{equation}
and $\mathbf{W_{m\times m}}$ is the diagonal matrix with $W(i,i)=w_i$.
  We show how to construct the  vectors $\mathbf{u_i}$ for $i=1,...,m$ in (\ref{twice_sparsifier}), which will suffice for proving the existence of linear-sized spectral sparsifiers for directed graphs. (\ref{formula_laplacian2}) allows writing  the undirected Laplacian after symmetrization as $\mathbf{L_{G_u}=B^\top W CC^\top WB}$ and setting $\mathbf{W_{o}=W CC^\top W}$. Since $\mathbf{W_{o}}$ is an  SPS matrix, we can always construct a $\mathbf{U}$ matrix with $\mathbf{u_i}$ for $i=1, ..., m$ as its column vectors:
 \begin{equation}
 \mathbf{
U_{n\times m}=\mathbf{[u_1,...,u_m]}={L_{G_u}^{\nicefrac{+}{2}}}B ^\top W_{o}^{\nicefrac{1}{2}}}.
 \end{equation}
 $\mathbf{U}$ contains all the information of the directed edges in $G$. It can be shown that $\mathbf{U}$ satisfies the following equation: 

\begin{equation}
\begin{split}
\mathbf{
UU^\top} & =\mathbf{\sum_{i} u_i {u_i}^\top}=\mathbf{{L_{G_u}}^{\nicefrac{+}{2}} B^\top W_{o} B  {L_{G_u}}^{\nicefrac{+\top}{2}}}\\
& \mathbf{={L_{G_u}}^{\nicefrac{+}{2}} L_{G_u}{L_{G_u}}^{\nicefrac{+\top}{2}}}\mathbf{={
\begin{bmatrix}
    I_{r\times r}    \\
         & 0_{(n-r)\times(n-r)}
\end{bmatrix}}},
\end{split}
\end{equation}
where $r$ is the rank of the $\mathbf{L_{G_u}}$.
According to Lemma \ref{theorem:dual_set}, we can always construct a diagonal matrix $\mathbf{T \in \mathbb{R}^{{m\times m}}}$ with $t_i$ as its $i$-th diagonal element. Then there will be at most $O(n/\epsilon ^2)$  positive diagonal elements in $\mathbf{T}$, which allows constructing $\mathbf{L_{S_u}=B^\top W_o^{\nicefrac{1}{2} }T W_o^{\nicefrac{1}{2} }B}$ that corresponds to the directed subgraph $S$ for  achieving $(1+\epsilon)$-spectral approximation of $G$ as required by (\ref{twice_sym}). It can be shown  that each $\mathbf{u_i}$ with a nonzero $t_i$ coefficient  corresponds to the outgoing edges pointed by the same node. Consequently, for directed graphs with bounded degrees, there will be $O(n/\epsilon ^2)$  total number of directed edges in the $(1+\epsilon)$-spectral sparsifier $S$.
\end{proof}
\section{A Practically-Efficient  Framework} \label{sec:practical}
Although we have shown that every directed graph has linear-sized spectral sparsifiers, there is no practically-efficient algorithm for constructing such sparsifiers. In this work, we exploit recent spanning-tree based spectral sparsification frameworks that have been developed for undirected graphs, and propose a practically-efficient algorithm  for spectral sparsification of   directed graphs.
\subsection{ Initial subgraph sparsifier}
 The idea of using  subgraphs as  preconditioners for more efficiently solving linear system of equations has been first introduced in \cite{vaidya1991solving},  showing that a maximum-spanning-tree (MST) subgraph can be leveraged as an $mn$-proxy of the original undirected graph. Recent nearly-linear time spectral sparsification algorithms for undirected graphs exploit similar ideas based on low-stretch spanning tree subgraphs \cite{zhuo:dac16,zhuo:dac18,elkin2008lst}. In this work, we exploit ideas that are closely related to spanning-tree based subgraphs as well as the Markov chains of random walks. The following procedure for constructing the initial subgraph sparsifiers for directed graphs has been developed: 
\begin{enumerate}
	\item Compute the transition matrix $\mathbf{P_{G_{sym}^{}} =  D_{G_{sym}}^{-1}A_{G_{sym}}}$ from $\mathbf{A_{G_{sym}}}=\mathbf{A_G}+{\mathbf{A^\top_G}}$, where, $\mathbf{L_{G_{sym}}}$ and  $\mathbf{D_{G_{sym}}}$ are Laplacian matrix and diagonal matrix for graph $G_{sym}$ respectively; 
    \item Construct an undirected graph $G_{sym}^{'}$ with $\mathbf{P_{G_{sym}^{}}}$ as its adjacency matrix, and find an MST subgraph $S^{mst}$ of $G_{sym}^{'}$.
    \item  Construct a directed subgraph $S^{0}$ according to the $S^{mst}$, and check every node in $S^{0}$: for each of the nodes that have at least one outgoing edge in $G$ but none in  $S^{0}$, include the outgoing edge  with the largest weight into $S^{0}$.
    \item Return the latest subgraph $S$ as the initial spectral sparsifier.
\end{enumerate}

Step 3)  will make sure that the graph Laplacians of directed graphs $G$ and  $S^{0}$ share the same rank and nullity. As aforementioned, if a node has more than one outgoing edge, in the worst case the neighboring nodes pointed by such outgoing edges will form a clique   in the corresponding undirected graph after symmetrization.  Consequently, when constructing the initial subgraphs from the original directed graphs, it is important to limit the number of outgoing edges for each node so that the resultant undirected graph after  Laplacian symmetrization will not be too dense. To this end,  emerging graph transformation techniques that allow   splitting high-degree nodes into multiple low-degree ones can be exploited. For example, recent research shows such split (e.g. uniform-degree tree) transformations can dramatically reduce graph irregularity while preserving critical graph connectivity, distance between node pairs, the minimal edge weight in the path, as well as outdegrees and indegrees when using push-based and pull-based vertex-centric programming \cite{nodehi2018tigr}.

 \subsection{Edge spectral sensitivity}
Denote the descending   eigenvalues and eigenvectors of $\mathbf{L^+_{S_u}}  \mathbf{L_{G_u}}$ by $\mu_{\max}={\mu _1} \ge {\mu _2} \ge \cdots \ge {\mu _n} > 0$ and $\mathbf{v_1, v_2,..., v_n}$, respectively.  In addition, let matrix  $\mathbf{V=[v_1, v_2,.. v_n]}$. Consider the following first-order generalized eigenvalue perturbation problem:
\begin{equation}\label{eqn:perturbation}
\mathbf{L_{G_u}(v_i+\delta v_i)=(\mu _i+\delta \mu _i)(L_{S_u}+\delta L_{S_u})(v_i+\delta v_i)},
\end{equation}
where a small perturbation $\delta\mathbf{L_{S_u}}$ in $\mathbf{S_u}$ is introduced and subsequently leads to the perturbed generalized eigenvalues and eigenvectors  $\mu _i+\delta \mu _i$ and $\mathbf{v_i+\delta v_i}$. Then the task of {spectral sparsification of general (un)directed  graphs} can be formulated as follows:  recover as few as possible  extra edges back to the initial   subgraph $S$ such that the largest  eigenvalues or the condition number of $\mathbf{L^+_{S_u}}  \mathbf{L_{G_u}}$ can be dramatically reduced. Expanding (\ref{eqn:perturbation}) by only keeping   the first-order terms simply leads to:
\begin{equation}\label{eqn:perturbation1}
\mathbf{L_{G_u}\delta v_i=\mu_i L_{S_u}\delta v_i +\delta \mu_i L_{S_u} v_i+\mu_i\delta L_{S_u} v_i,~ \delta v_i=\sum_{j=1}^{n}\xi_{ij}v_j}.
\end{equation}
Since both $\mathbf{L_{G_u}}$ and $\mathbf{L_{S_u}}$ are SPS matrices, $\mathbf{L_{S_u}}$-orthogonal generalized eigenvectors  $\mathbf{v_i}$ for $i=1,...,n$ can be found to satisfy:
\begin{equation} \label{eqn:LS orthogonal}
 \mathbf{v_i^\top L_{S_u} v_j}=\begin{cases}
    1, & \text{$i=j$}\\
    0, & \text{$i\neq j$}.
  \end{cases}
\end{equation}
By expanding $\mathbf{\delta L_{S_u}}$ with only the first-order terms,
the spectral perturbation for each off-subgraph edge can be expressed as:
\begin{equation}\label{perturb_digraph}
\frac{\delta {\mu _i}}{{\mu _i}} =-\mathbf{v_i^\top \delta L_{S_u} v_i}=-\mathbf{v_i^\top(\delta L_S L_S^\top+ L_S \delta  L_S^\top) v_i},
\end{equation}
where $\mathbf{\delta L_S}= \mathbf{{{w_{p,q}}e_{p,q}}\mathbf{e_p}^\top}$ for ${(p,q)\in E_G\setminus E_S}$, $\mathbf{e_p}\in \mathbb{R}^n$ denotes the vector with only the $p$-th element being $1$ and others being $0$, and  $\mathbf{e_{p,q}}=\mathbf{e_p}-\mathbf{e_q}$.
The \textbf{spectral sensitivity} $\delta {\mu s}_{p,q}$ for the off-subgraph edge $(p,q)$ can be computed by:
\begin{equation}\label{eqn:spec sensitity}
\delta {\mu s}_{p,q} =\mathbf{v_1^\top \delta L_{S_u} v_1}.
\end{equation}
  (\ref{eqn:spec sensitity}) allows computing the {spectral sensitivity} of the dominant generalized eigenvalue  with respect to the Laplacian perturbation due to adding each extra off-subgraph edge into  $S$, which thus can be leveraged to rank the spectral importance of each edge. As a result, spectral sparsification of general (un)directed graphs can be achieved by only recovering the top few   off-subgraph edges that have the  largest spectral sensitivities  into  $S$. 

Since the above framework is based on spectral matrix perturbation analysis, compared to existing spectral   sparsification methods that are limited to specific types of graphs, such as undirected graphs or strongly-connected directed graphs \cite{cohen2017almost,cohen:focs18}, the proposed  spectral graph sparsification framework is more universal and  thus will be applicable to a much broader range   of graph problems. 

\subsection{Approximate dominant eigenvectors  }
 A generalized power iteration   method is proposed  to allow much faster computation of the dominant generalized eigenvectors for spectral   sparsification of directed graphs.  Starting from any initial random vector expressed as $\mathbf{ h_0=\sum_{i}\alpha_iv_i}$, the  dominant generalized eigenvector  $\mathbf{v_1}$ can be approximately computed  by performing the following $t$-step power iterations:
\begin{equation} \label{eqn:ht}
\mathbf{v_1\approx h_t=(L_{S_u}^{+}L_{G_u})^th_0}=\sum_{i} \alpha_i\mu_i^t \mathbf{v_i}.
\end{equation}
When the number of power iterations is small (e.g., $t\le3$),  $\mathbf{h_t}$ will be a linear combination of the first few dominant generalized eigenvectors corresponding to the  largest few eigenvalues.
Then the  { spectral sensitivity}   for the off-subgraph edge $(p,q)$ can be approximately computed by
\begin{equation} \label{eqn:Q}
\begin{split}
&\delta {\mu s}_{p,q}  \approx \mathbf{h_t^\top \delta L^{}_{S_u} h_t}
\end{split},
\end{equation}
which will allow us to well approximate the spectral sensitivity in (\ref{perturb_digraph}) for ranking off-subgraph edges during spectral sparsification. The key to fast computation of $\mathbf{h_t}$ using generalized power iterations is to quickly solve the linear system of equations $\mathbf{L_{S_u} x=b}$, which   requires to explicitly construct   $\mathbf{L_{S_u}}$  rather than $\mathbf{L_{G_u}}$.   To this end, we leverage the latest Lean Algebraic Multigrid (LAMG) algorithm that is capable of handling the undirected graphs with negative edge weights as long as the Laplacian matrix is SPS.  The LAMG algorithm also enjoys an empirical $O(m)$ runtime   complexity for solving large scale graph Laplacian matrices \cite{livne2012lean}.

 \subsection{Lean algebraic multigrid (LAMG)}
 The setup phase of LAMG contains two main steps:
First, a nodal elimination procedure is performed to eliminate disconnected and low-degree nodes. Next, a node aggregation procedure is applied for aggregating strongly connected nodes   according to the following   affinity metric $c_{uv}$    for nodes $u$ and $v$:
 \begin{equation}\label{eqn:agg}
  c_{uv} = \frac{{\|(X_u, X_v)\|}^2}{(X_u, X_u)(X_v, X_v)},\;\;\;   (X,Y) = \Sigma_{k=1}^{K}{x^{(k)}  \cdot y^{(k)}}
\end{equation}
where $X_u = (x_{u}^{(1)}\dots x{_u}{^{(K)}})$ is computed by applying a few Gauss-Seidel (GS) relaxations using $K$ initial random vectors to the linear system equation $\mathbf{L_{G_u} x}=0$. Let $\mathbf{\tilde{x}}$ represent the approximation of the true solution $\mathbf{x}$ after applying several GS relaxations to $\mathbf{L_{G_u} x}=0$. Due to the smoothing property of GS relaxation, the latest error can be expressed as $\mathbf{e_s \,=\,x-\tilde{x}}$, which will only contain the smooth components of the initial error, while the highly oscillating modes   will be effectively damped out \cite{briggs2000multigrid}.  It has been shown that the node affinity metric $c_{uv}$ can   effectively reflect the distance or strength of connection between nodes in a graph: a larger $c_{uv}$ value indicates a stronger connection between nodes $u$ and $v$ \cite{livne2012lean}. Therefore, nodes $u$ and   $v$ are considered strongly connected   to each other if 
$\mathbf{x_u}$ and $\mathbf{x_v}$ are highly correlated for all the $K$ test vectors, which thus should  be aggregated to form a coarse level node. 

Once the multilevel hierarchical representations of the original graph (Laplacians) have been created according to the above scheme, algebraic  multigrid (AMG) solvers can be built and subsequently leveraged to solve large Laplacian matrices efficiently.

\subsection{Edge spectral similarities}
The proposed spectral sparsification algorithm will first sort all off-subgraph edges according to their spectral sensitivities in  descending order $({p_1,q_1}),({p_2,q_2}),...$ and then select top few off-subgraph edges to be recovered to the initial subgraph.  To avoid recovering  redundant edges into the subgraph, it is indispensable to check the edge spectral similarities: only the edges that are not similar to each other will be added to the initial sparsifier.
 To this end, we exploit the following spectral embedding of off-subgraph edges  using   approximate dominant generalized eigenvectors $\mathbf{h_t}$ computed by (\ref{eqn:ht}): 
 \begin{equation}\label{eqn:psi}
 \psi_{p,q}(h_t)=\sum_{k} w_{p,q_k} \mathbf{h_t^\top(e_{p,q}e_{p,q_k}^\top +e_{p,q_k}e_{p,q}^\top)h_t},
 \end{equation}
where  $(p,q_k)$ are the directed edges  sharing the same head with $(p,q)$ but different tails. Then the proposed scheme for checking spectral similarity of two off-subgraph edges will include the following steps:
\begin{enumerate}
 \item Perform $t$-step power iterations with $r=O(\log n)$ initial random vectors  $\mathbf{h^{(1)}_0,...,h^{(r)}_0}$  to compute $r$ approximate dominant generalized eigenvectors  $\mathbf{h^{(1)}_t,...,h^{(r)}_t}$; 
 \item For each edge $(p,q)$, compute   a $r$-dimensional spectral embedding vector $\mathbf{s_{p,q}}  \in \mathbb{R}^r$ with $s_{p,q}(r)=\psi_{p,q}(h_t^{(r)})$;
 \item Check the spectral similarity of two off-subgraph edges $(p_i,q_i)$ and $(p_j,q_j)$ by computing
 \begin{equation}
\textrm{SpectralSim} =1-\frac{||\mathbf{s_{p_i,q_i}}-\mathbf{s_{p_j,q_j}}||}{\max(||\mathbf{s_{p_i,q_i}}||_2,||\mathbf{s_{p_j,q_j}}||)}.
\end{equation}
 \end{enumerate}
 
 If $\textrm{SpectralSim} < \epsilon$ for a given $\epsilon$,   edge$(p_i,q_i)$  is considered spectrally dissimilar with $(P_j,q_j)$. Algorithm flow Algorithm \ref{alg:edge_similarity} is presented to show the edge similarity checking for a list of off-subgraph edges.
 \begin{algorithm}[!htbp]
\small { \caption{\textrm{Edge Similarities Checking}} \label{alg:edge_similarity}
\textbf{Input:} $E_{\textrm{list}}$, $\mathbf{L}_{G}$, $\mathbf{L}_{S}$,  $d_{{\textrm{out}}}$,  $\epsilon$ \\
    \begin{algorithmic}[1]
    \STATE{Perform t-step power iterations with $r=O(\log n)$ initial random vectors  $\mathbf{h^{(1)}_0,...,h^{(r)}_0}$  to compute $r$ approximate dominant generalized eigenvectors  $\mathbf{h^{(1)}_t,...,h^{(r)}_t}$};
     \STATE{Choose each edge $(p,q)$  whose starting node has out-degree less than $d_{\textrm{out}}$ into a new $E_\textrm{list}$};
     \STATE{Compute  a $r$-dimensional  edge similarity vector $\mathbf{s_{p,q}}  \in \mathbb{R}^r$ for $\forall (p,q)\in E_\textrm{list}$: $s_{p,q}(r)=\psi_{p,q}(h_t^{(r)})$ };
     \STATE{let $E_{\textrm{addlist}}=[(p_1,q_1)]$};
     \FOR{i=1:$E_{\textrm{list}}$}
  
     \IF{$1-\frac{||\mathbf{s_{p_i,q_i}}-\mathbf{s_{p_j,q_j}}||}{\max(||\mathbf{s_{p_i,q_i}}||,||\mathbf{s_{p_j,q_j}}||)}<\epsilon$, for $ \forall    (p_j,q_j)\in E_{\textrm{addlist}}$}
    \STATE{$E_{\textrm{addlist}}=[E_{\textrm{addlist}};(p_i,q_i)]$};
    \ENDIF

    \ENDFOR
     \STATE {Return graph $E_{\textrm{addlist}}$ };
    \end{algorithmic}
    }
    
\end{algorithm}

\subsection{Algorithm flow and complexity}
The algorithm flow for directed graph spectral sparsification is described in Algorithm \ref{alg:directed_graph_spar}, while its complexity has been summarized as follows:
 \begin{enumerate}
 \item[\textbf{(a)}] Generate an initial subgraph $S$ from the original directed graph in $O(m \log n)$ or $O(m + n \log n)$ time;
 \item[\textbf{(b)}] Compute the approximate dominant eigenvector $\mathbf{h_t}$  and  the spectral sensitivity of each off-subgraph edge in $O(m)$ time;
  \item[\textbf{(c)}] Recover a small amount of spectrally-dissimilar  off-subgraph edges into the latest subgraph $S$ according to their spectral sensitivities and similarities in $O(m)$ time; 
 \item[\textbf{(d)}] Repeat steps \textbf{(b)} and \textbf{(c)} until  the desired condition number or spectral similarity is  achieved. 
 \end{enumerate}

\begin{algorithm}[!htbp]
\small { \caption{Algorithm Flow for Directed Graph Sparsification} \label{alg:directed_graph_spar}
\textbf{Input:}  $\mathbf{L}_{G}$, $\mathbf{L}_{S}$,  $d_{{\textrm{out}}}$, $\textrm{iter}_{\textrm{max}}$, $\mu_{\textrm{limit}}$, $\alpha$, $\epsilon$ \\
    \begin{algorithmic}[1]
    
     \STATE{Calculate largest generalized eigenvector $h_t$, largest generalized eigenvalue $\mu_{\textrm{max}}$ and let $\mu_{\textrm{max}}=\mu_{\textrm{in}}$, $ \textrm{iter}=1$};
    \WHILE{ $ \mu_{\textrm{max}}<\mu_{\textrm{limit}}$, $\textrm{iter}<{\textrm{iter}}_{\textrm{max}}$}
   
    \STATE{  Calculate the spectral sensitivities $\delta {\mu s}_{p,q}$ for each off-subgraph edges $(p,q)\in E_{G\backslash S}$};
    \STATE{Sort spectral sensitivities in descending order and obtain  the top $\alpha\%$ off-subgraph edges into  edge list $E_{\textrm{list}}=[(p_1,q_1),(p_2,q_2),...]$};
    
    \STATE{ Do $E_{\textrm{addlist}}=\textrm{EdgeSimilaritiesChecking}(E_{\textrm{list}},L_G,L_S,d_{\textrm{out}}\epsilon)$ };
    \STATE{Update $S_{\textrm{new}}=S+e_{\textrm{addlist}}$ and calculate largest generalized eigenvector $h_{t\textrm{new}}$, largest generalized eigenvalue $\mu_{\textrm{maxnew}}$ based on $L_G$ and $L_{S\textrm{new}}$ };
    \IF{$\mu_{\textrm{maxnew}}<\mu_{\textrm{max}}$}
    \STATE{Update $S=S_{\textrm{new}}$, $h_t=h_{t\textrm{new}}$, $\mu_{\textrm{max}}=\mu_{\textrm{maxnew}}$ };
    \ENDIF

    \STATE{$\textrm{iter}=\textrm{iter}+1$};
    \ENDWHILE
     \STATE {Return graph $S$ and $\mathbf{L}_S$, $\mu_{\textrm{max}}$, $\frac{|E_S|}{|E_G|}$, $\frac{\mu_{\textrm{in}}}{\mu_{\textrm{max}}}$}.
    \end{algorithmic}
    }
    
\end{algorithm}

  \section{APPLICATIONS of Directed Graph Sparsification}\label{sec:application}
  Spectral graph sparsification algorithms can be potentially applied to accelerate many graph  and numerical algorithms \cite{cohen2017almost}. In this work, we demonstrate the applications of the proposed  sparsification algorithm in  solving directed Laplacian problems (e.g., $Lx=b$),   \cite{cohen2016faster,cohen2017almost}, computing personalized PageRank vectors \cite{cohen:focs18}, as well as spectral graph partitioning. 

  \subsection{PageRank and personalized PageRank}
  The idea of PageRank is to give a measurement of the importance for each web page. For example, PageRank algorithm aims to find the most popular web pages, while the personalized PageRank algorithm aims to find the pages that users will most likely to visit. To  state it mathematically, the PageRank vector $\mathbf{p}$ satisfies the following equation:
 \begin{equation}\mathbf{
    \mathbf{p}=A_G^\top D_G^{-1}\mathbf{p}},
 \end{equation}
  where $\mathbf{p}$ is also the eigenvector of $\mathbf{A_G ^\top D_G^{-1}}$ that corresponds to the eigenvalue equal to $\mathbf{1}$. Meanwhile, $\mathbf{p}$ represents the stable distribution of random walks on graph $G$. However, $\mathbf{D_G^{-1}}$ can not be defined if there exists nodes that have no outgoing edges.  To deal with such situation,  a self-loop  with a  small edge weight can be added  for each node. 
  
  The stable distributions of (un)directed  graphs may not be unique. For example, the undirected graphs that have multiple  strongly-connected components, or  the directed graphs that have    nodes  without any outgoing edges, may have non-unique distributions.  In addition, it may take very long time for a random walk to converge to a stable distribution on a given (un)directed graph.
  
  To avoid such situation in PageRank, a jumping factor $\alpha$ that describes the  possibility at $\alpha$ to jump to a uniform vector can be added, which is shown as follows: 
  \begin{equation}\mathbf{
      {p}=(1-\alpha) A_G^\top D_G^{-1}{p} + \frac{\alpha}{n}{1}},
  \end{equation}
  
 \begin{equation}\mathbf{
     {p}= \frac{\alpha}{n}(I-(1-\alpha) {A_G}^\top D_G^{-1})^{-1}{1} },
  \end{equation}
  where $\alpha \in [0, 1]$ is a jumping constant. After applying Taylor expansions,  we can obtain that $\mathbf{{p}= \frac{\alpha}{n}\sum_{i}{((1-\alpha) A_G^\top D_G^{-1})^i}}$ . By setting the proper value of $\alpha$ (e.g., $\alpha=0.15$), the term $(1-\alpha)^i$ will be quickly reduced with increasing $i$.  Instead of starting with a uniform vector$\frac{\alpha}{n}$, a  nonuniform personalization vector $\mathbf{p}r$ can be applied:
  \begin{equation}\mathbf{
      {p}=(1-\alpha) A_G^\top D_G^{-1}{p} +\alpha pr}.
  \end{equation}
  In this work,  we show that the PageRank vector obtained with the sparsified graph can preserve the  original PageRank information. After obtaining the PageRank vector computed using the sparsifier, a few GS relaxations  will be applied to further improve the solution quality.

  \subsection{Directed Laplacian  solver}
  Consider  the solution of the following linear systems of equations:
  \begin{equation}\label{eqn:linear system}
    \mathbf{  L x=b}.
  \end{equation}
  Recent research has been focused on more efficiently solving the above problem when $\mathbf{L}$ is a Laplacian  matrix of an undirected graph \cite{kelner2014almost,koutis2010approaching}. In this work, we will mainly focus on solving  nonsymmetric Laplacian matrices that correspond to directed graphs.
  \begin{lemma}\label{lemma:directed solver}
  When solving   (\ref{eqn:linear system}), the right preconditioning system is applied, leading to the following alternative linear system of equations:
  \begin{equation}\label{eqn:directed linear}
     \mathbf{ L_{G_u} y=b},
  \end{equation}
  where vector $\mathbf{b}$  will lie in the left singular vector space. When the solution of  (\ref{eqn:directed linear}) is obtained, the solution of    (\ref{eqn:linear system}) is given by $\mathbf{ L_G^{\top}y=x}$.
  
  \end{lemma}
     It is  obvious that solving the above equation is equivalent to  solving the problem of $\mathbf{L_GL_G^\top L_G^{+ \top} x=b}$.  In addition, $\mathbf{L_{G_u}}$ is a Laplacian matrix of an undirected graph that can be much denser than $\mathbf{L_{G}}$.
Therefore, we propose to solve the linear system of $\mathbf{L_{S_u} \tilde{y}=b }$  instead to effectively approximate  (\ref{eqn:directed linear}) since   $G_{S_u}$ is sparser  than   $G_{G_u}$ and more efficient to solve in practice.

We analyze the solution errors based on the generalized eignvalue problem of $\mathbf{L_{G_U}}$ and $\mathbf{L_{S_U}}$. We have $\mathbf{ V L_{G_U}V^\top=\mu}$ and $\mathbf{ V L_{S_U}V^\top=I}$, where $\mathbf{V=[v_1,v_2,..v_n]}$, $\mathbf{\mu}$ is the diagonal matrix with its generalized eigenvalues $\mu_i \ge 1$ on its  diagonal.
Since the errors can be calculated  from the following procedure: 
\begin{equation}
\begin{split}
 & \mathbf{L_{G_U}y-L_{S_U}\tilde{y} =  L_{G_U}(y-\tilde{y})+(L_{G_U}-L_{S_U})\tilde{y} =0}
    \end{split},
\end{equation}
 we can write the error term as follows:
\begin{equation}\mathbf{
    (\tilde{y}-y)=L_{G_U}^+(L_{G_U}-L_{S_U})\tilde{y}}.
\end{equation}
Since $\mathbf{\tilde{y}=\sum_i a_i v_i}$, the error can be further expressed as 
\begin{equation}\label{eqn:error}
    \mathbf{(\tilde{y}-y)=\sum_i a_i(1-\frac{1}{\mu_i})v_i}.
\end{equation}
Therefore, the error term (\ref{eqn:error}) can be generaly considered as a combination of high-frequency errors (generalized eigenvectors with respect to high generalized eigenvalues) and low-frequency errors (generalized eigenvectors with respect to low generalized eigenvalues). After applying  GS relaxations,  the high-frequency error terms can be efficiently removed (smoothed), while the low-frequency errors tend to become zero if the generalized eigenvalues approach $1$ considering  $(1-\frac{1}{\mu_i})$  tends to be approaching zero. As a result, the  error can be effectively eliminated using the above solution smoothing procedure.

 In summary, in the proposed directed Laplacian solver, the following steps are needed: 
\begin{itemize}
    \item[\textbf{(a)}]   We will first extract a spectral sparsifier  $\mathbf{L_S}$ of a given (un)directed graph $\mathbf{L_G}$. Then, it is possible to compute an approximate solution   by exploiting  its spectral sparsifier $\mathbf{L_{S_u}=L_S L_S^\top}$ via solving $\mathbf{\tilde{y}=L_{S_u}^+b}$ instead. 
    
    \item[\textbf{(b)}] Then we improve the approximate solution  $\mathbf{\tilde{y}}$ by getting rid of the high-frequency errors via applying a few steps of GS iterations \cite{multigrid:book}. 
    \item[\textbf{(C)}] The final solution is obtained from $\mathbf{x=L_G^\top \tilde{y}}$.
\end{itemize}

  \subsection{ Directed graph partitioning}
  It has been shown that  partitioning and clustering of directed graphs can play very important roles in a variety of applications related to machine learning \cite{malliaros2013clustering}, data mining and circuit synthesis and optimization \cite{micheli1994synthesis}, etc. However, the efficiency of existing methods for partitioning directed graphs strongly depends on the complexity of the underlying graphs \cite{malliaros2013clustering}. 
  
  In this work, we propose  a spectral method for directed graph partitioning problems. For an undirected graph, the eigenvectors corresponding to the first few smallest eigenvalues can be utilized for the spectral partitioning purpose \cite{spielmat1996spectral}. For a directed graph $G$ on the other hand, the left singular vectors  of  Laplacian $L_G$ will be required for directed graph partitioning. The eigen-decomposition  of its symmetrization $\mathbf{L_{G_{U}}}$ can be wirtten as 
  \begin{equation}
     \mathbf{ L_{G_{U}}=\sum_i \mu_{i} v_{i}v_{i}^\top},
  \end{equation}
  where $0=\mu_{1} \le...\mu_{k}$ and $\mathbf{v_{1},...,v_{k}}$, with $k\le n $ denote the Laplacian eigenvalues and eigenvectors, respectively. There may not be $n$ eigenvalues if  when there are some nodes without any outgoing edges.  In addition, the spectral properties of $L_{G_{U}}$ are more complicated since the eigenvalues always have multiplicity (either algebraic or geometric multiplicities).  For example, the eigenvalues according to the symmetrization of  the directed graph in Figure \ref{fig:partition} have a a few multiplicities: $\mu_2=\mu_3$, $\mu_4=\mu_5=\mu_6=\mu_7$, $\mu_9=\mu_{10}$.
  
 Therefore, we propose to exploit the eigenvectors  (left singular vectors of directed Laplacian)  corresponding to the first few different eigenvalues (singular values of directed Laplacian) for directed graph partitioning. For example, the partitioning result of directed graph in Figure \ref{fig:partition} will depend on the eigenvectors of $v_1,v_2,v_4,v_8$ that correspond to eigenvalues of $\mu_1,\mu_2,\mu_4,\mu_8$.  As shown in Fig (\ref{fig:partition}), the spectral partitioning results can quite different between the directed and undirected graph with the same set of nodes and edges.
 \begin{figure}[H]
\centering 
\includegraphics[scale=0.425]{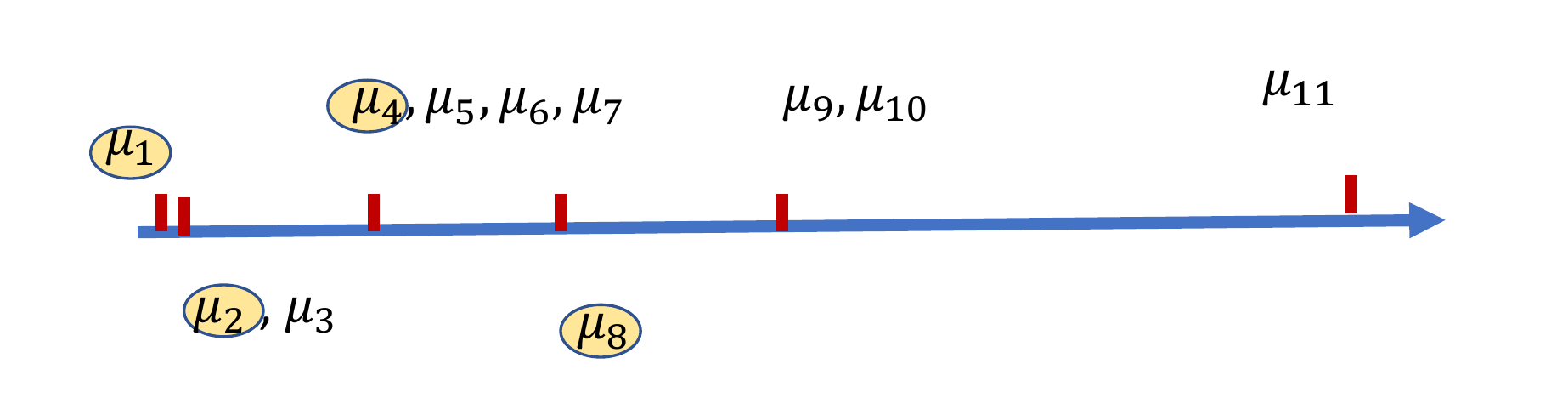}
\caption{Eigenvalues distribution of $L_{G_U}$  for the directed graph in Figure \ref{fig:partition}
\protect\label{fig:eig_distr}}
\vspace{0.1cm}
\end{figure} 
In general, it is possible to first extract a spectrally-similar directed graph before any of the prior  partitioning algorithms are applied. Since the proposed spectral sparsification algorithm   can well preserve the structural (global) properties of the original graphs, the   partitioning results obtained from the sparsified graphs will be very similar to the original ones. 



 \begin{figure}
\centering 
\includegraphics[width=0.225\textwidth]{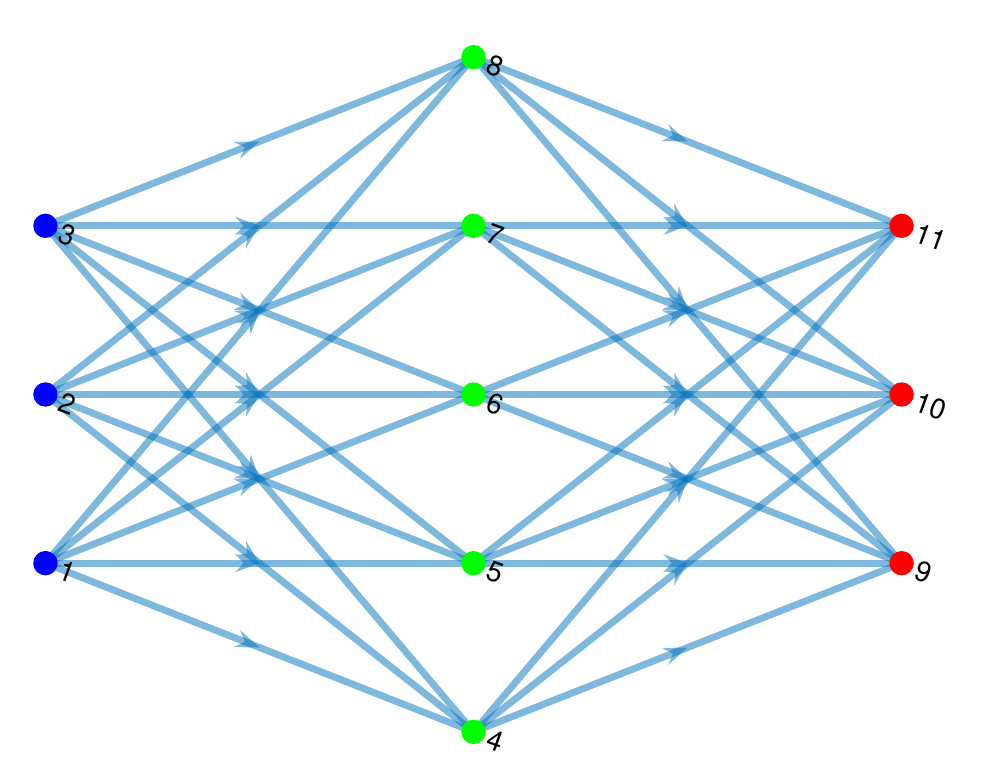}
\includegraphics[width=0.225\textwidth]{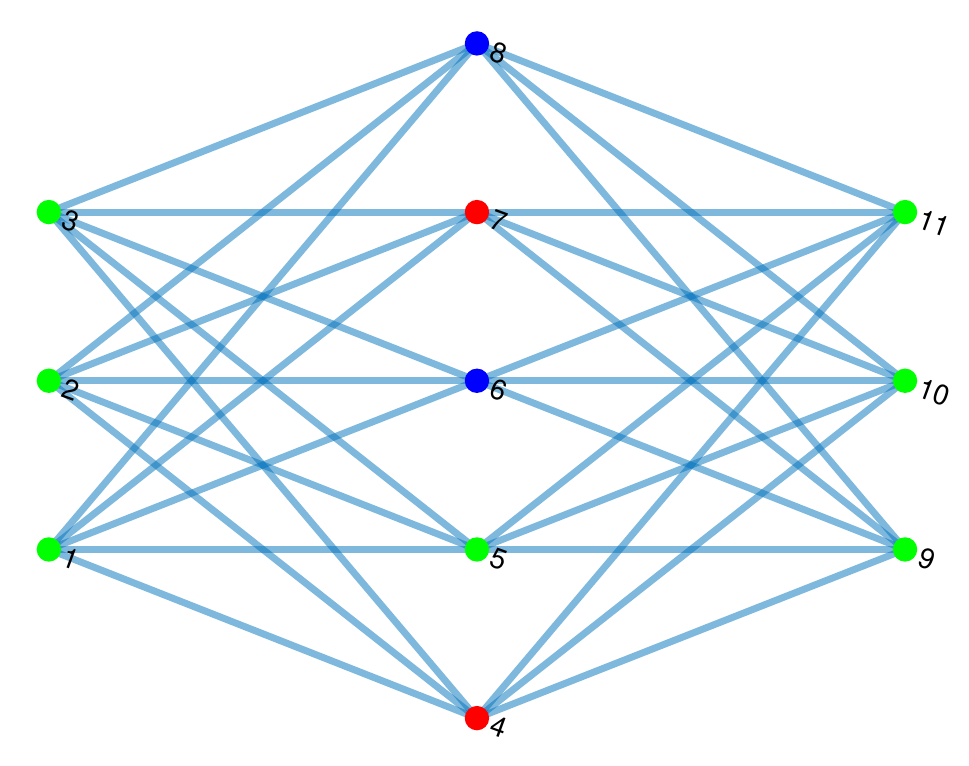}\caption{ Spectral partitioning of directed (left) and undirected graphs (right). The nodes within the same cluster are assigned the same color.
\protect\label{fig:partition}}
\vspace{0.15cm}
\end{figure}

  \vspace{-0.15cm}

\section{Experimental Results}\label{sec:result}
\begin{figure}
\centering \epsfig{file=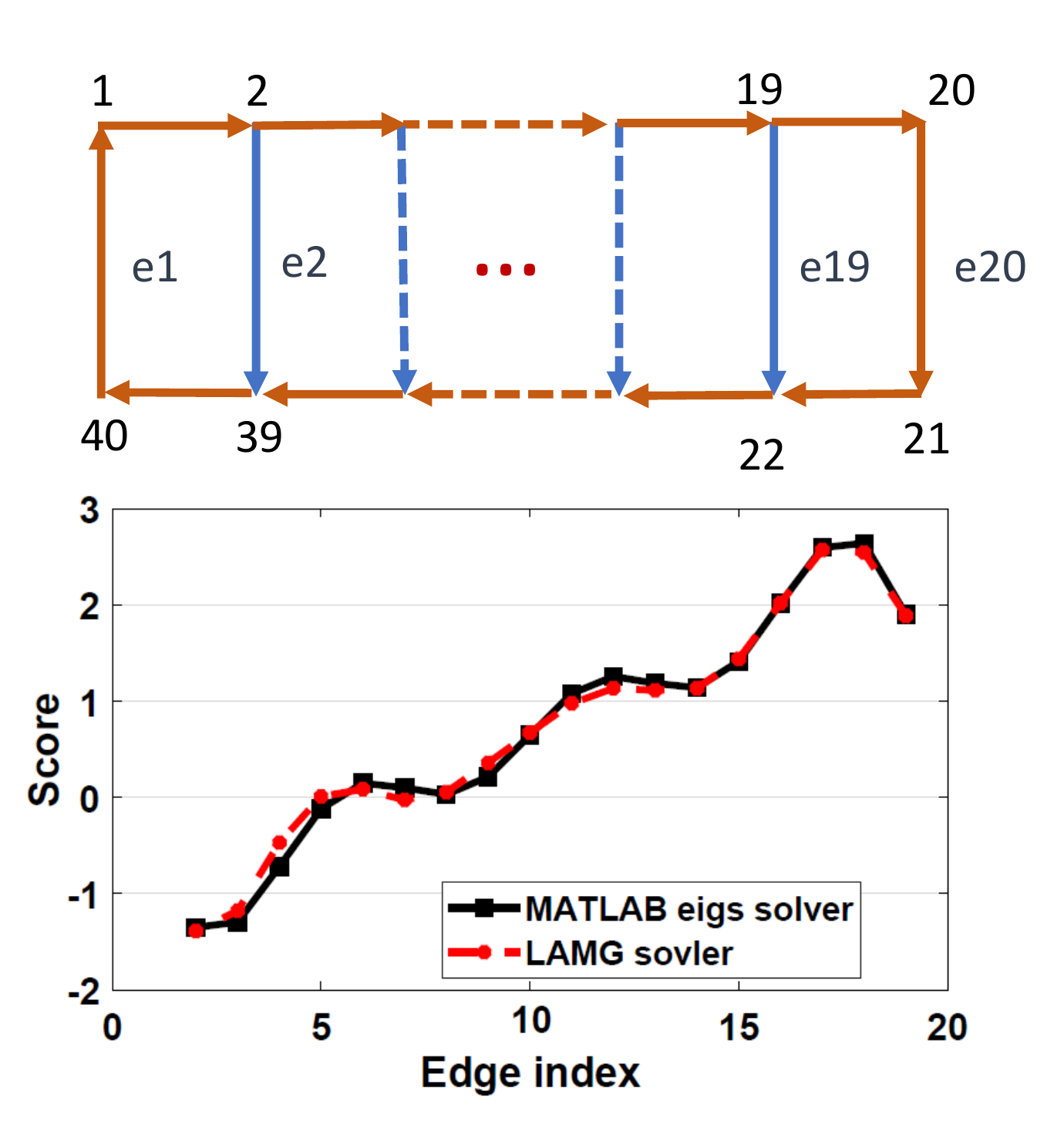, scale=0.2995060} \epsfig{file=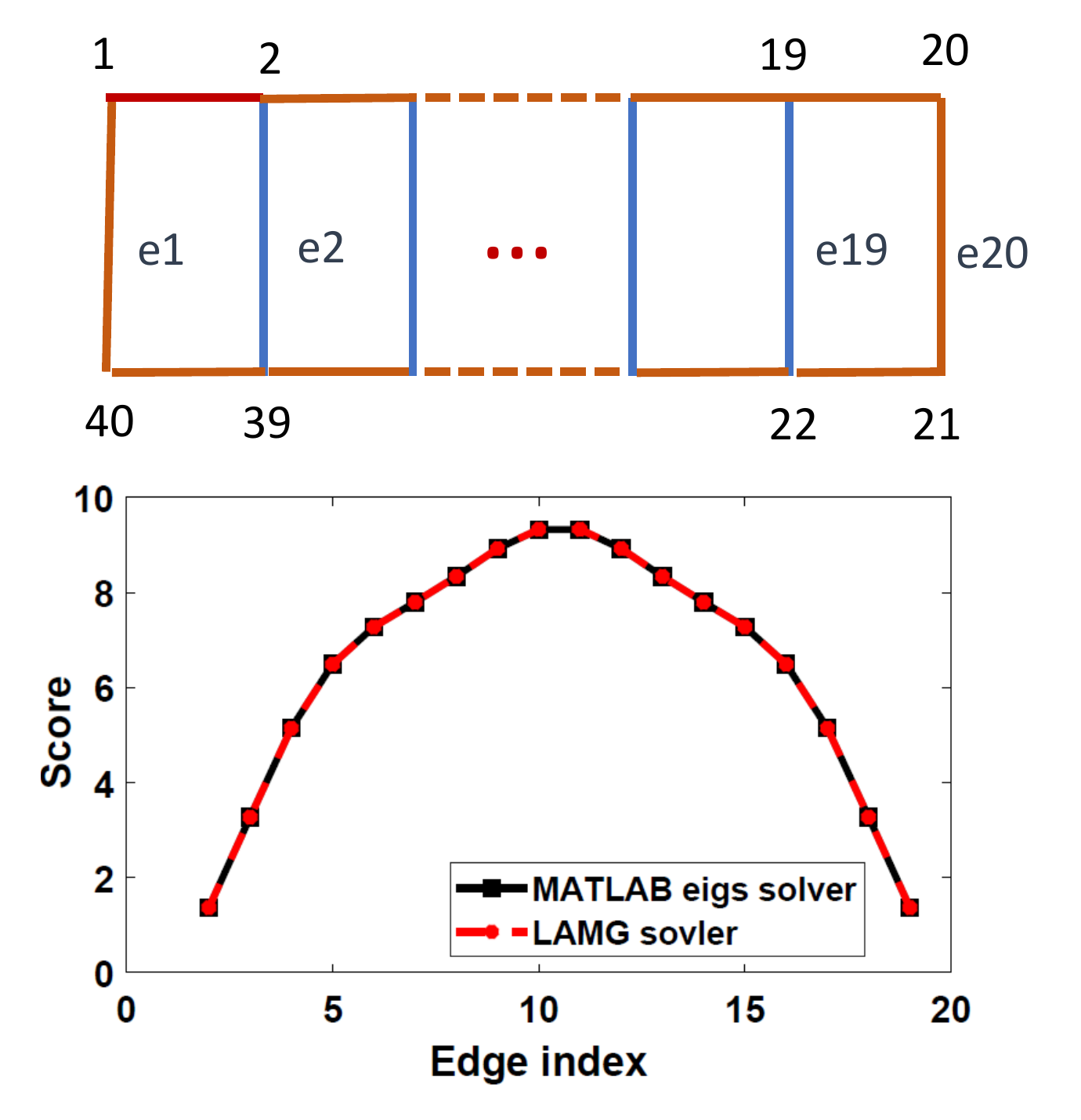, scale=0.2995060} 
\caption{The spectral sensitivities (scores) of off-subgraph edges ($e2$ to $e19$) for the directed  (left)  and   undirected graph (right).\protect\label{fig:score}}
\end{figure}
The proposed algorithm for spectral sparsification of directed graphs has been implemented using MATLAB and C++ \footnote{Code available at: https://github.com/web1799/directed$\_$graph$\_$sparsification}. Extensive experiments have been conducted to evaluate the proposed method with various types of directed graphs obtained from public-domain data sets \cite{davis2011matrix}. 


Figure \ref{fig:score} shows the spectral sensitivities of all the off-subgraph edges ($e2$ to $e19$) in both directed and undirected graphs calculated using MATLAB's ``eigs" function and the proposed method based on (\ref{eqn:Q}) using the LAMG solver, respectively. The edges of subgraphs for both directed and undirected graphs are represented with red color. The spectral sensitivities of all the off-subgraph edges ($e2$ to $e19$ in blue color) with respect to the dominant eigenvalues ($\mu_{max}$ or $\mu_{1}$) in both directed and undirected graphs are plotted. We observe that spectral sensitivities for directed and undirected graphs are drastically different from each other. The reason is that the spectral sensitivities for off-subgraph edges in the directed graph depend on the edge directions. It is also observed that the approximate spectral sensitivities calculated by the proposed $t$-step power iterations with the LAMG solver match the true solution very well for both directed and undirected graphs.

Table \ref{table:fiedler} shows more comprehensive results on directed graph spectral sparsification for a variety    of real-world directed graphs using the proposed method, where $|V_G|(|E_G|)$ denotes the number of nodes (edges) for the original directed graph $G$; $|E_{S^0}|$ and $|E_S|$ denote the numbers of edges in the initial   subgraph $S^0$ and final spectral sparsifier $S$. Notice that we will directly apply the Matlab's "eigs" function if the size of the graph is relatively small ($|E_{S^0}|<1E4$); otherwise we will switch to the LAMG solver for better efficiency when calculating the approximate generalized eigenvector $\mathbf{h_t}$. We report the total runtime for the eigsolver using either the LAMG solver or "eigs" function. $\frac{\mu_{\textrm{in}}}{\mu_{\textrm{max}}}$ denotes the reduction rate of the largest generalized eigenvalue of $\mathbf{L^+_{S_u}}  \mathbf{L_{G_u}}$. We also plot the detailed reduction rates of the largest generalized eigenvalue when adding different number of off-subgraph edges to the sparsifiers of graph ``gre$\_$115" and ``peta" in Figure \ref{fig:reduction_ratio}. It shows that the largest generalized eigenvalue can be effectively reduced if sufficient off-subgraph  edges are included into the sparsifier.


\begin{table}
\begin{center}

 \addtolength{\tabcolsep}{-2.5pt} \centering
\caption{Results  of directed graph spectral sparsification }
\label{table:result1}

\begin{tabular}{|c|c|c|c|c|c|c|}
 \hline   Test Cases & $|V_G|$ & $|E_G|$ & $\frac{|E_{S^{0}}|}{|E_G|}$& $\frac{|E_{S}|}{|E_G|}$ & time (s)&$\frac{\mu_{\textrm{in}}}{\mu_{\textrm{max}}}$  \\
  \hline gre\_115 & 1.1E2 &  4.2E2 & 0.46 &0.71  &0.05&8.2E4X \\
 \hline gre\_185 & 1.8E2 & 1.0E3 & 0.33& 0.46 & 0.14 &9.8E3X  \\
 \hline harvard500& 0.5E3 & 2.6E3&  0.31& 0.40& 0.64 &1.2E3X \\
  \hline cell1 & 0.7E4 &  3.0E4 & 0.31 &0.57  &3.10 &1.0E5X \\
  \hline pesa& 1.2E4 & 8.0E4 & 0.27& 0.51& 8.80 & 5.3E8X  \\
\hline big & 1.3E4 & 0.9E5 & 0.27&0.49& 12.86& 4.1E11X\\
   \hline gre\_1107 & 1.1E3 &  5.6E3 & 0.26 &0.39& 0.24 & 58X \\
 \hline wordnet3 & 0.8E5 & 1.3E5 & 0.64 &0.84 & 50.00 & 12X  \\
 \hline p2p-Gnutella31 & 0.6E5 & 1.5E5 & 0.35 &0.43 & 11.90 & 6X  \\
   \hline p2p-Gnutella05 & 8.8E3 & 3.2E4 & 0.23 &0.65 & 27.59 & 35X  \\
  \hline mathworks100 &1.0E2 &5.5E2 & 0.20&0.50 & 0.04 & 30X \\
  \hline ibm32 & 3.2E1 & 1.3E2 & 0.46&0.57 & 0.02 & 12X \\
  \hline
\end{tabular}
\label{table:fiedler}
\end{center}
\end{table}

 \begin{table}[ht]
 \begin{center}

 \addtolength{\tabcolsep}{-3.7pt} \centering
 \caption{ the errors between the exact and approximate  solutions of $L_G x=b$ with or without Gauss Seidel smoothing }
 \label{table:result2}
 \begin{tabular}{|c| c |c| c| c| c| c| c| c|c|  }
  \hline   Test Cases & gre$\_$115& gre$\_$185 & cell1 & pesa & big& gre$\_$1107& wordnet3 \\
  \hline w/o smooth. & 0.41& 0.42& 0.44 & 2.1E-4&4.3E-3&0.6&0.72\\
   w/ smooth. & 0.04&0.12&0.07&8.0E-9&1.1E-4&0.10&0.07\\
 \hline

 \end{tabular}
 \label{table:fiedler1}
 \end{center}
 \end{table}
 

\begin{figure}[H]
\centering 
\includegraphics[width=0.235\textwidth]{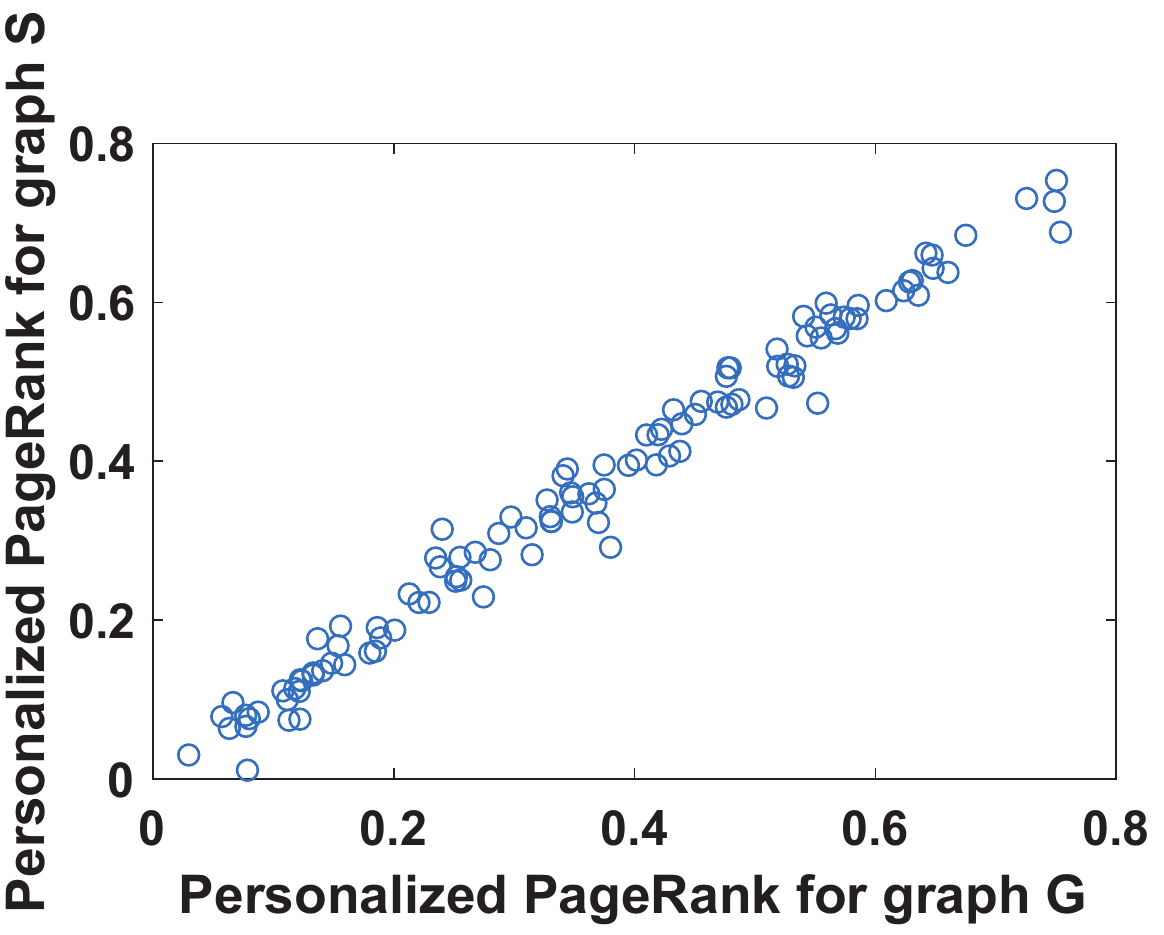}
%
\includegraphics[width=0.235\textwidth]{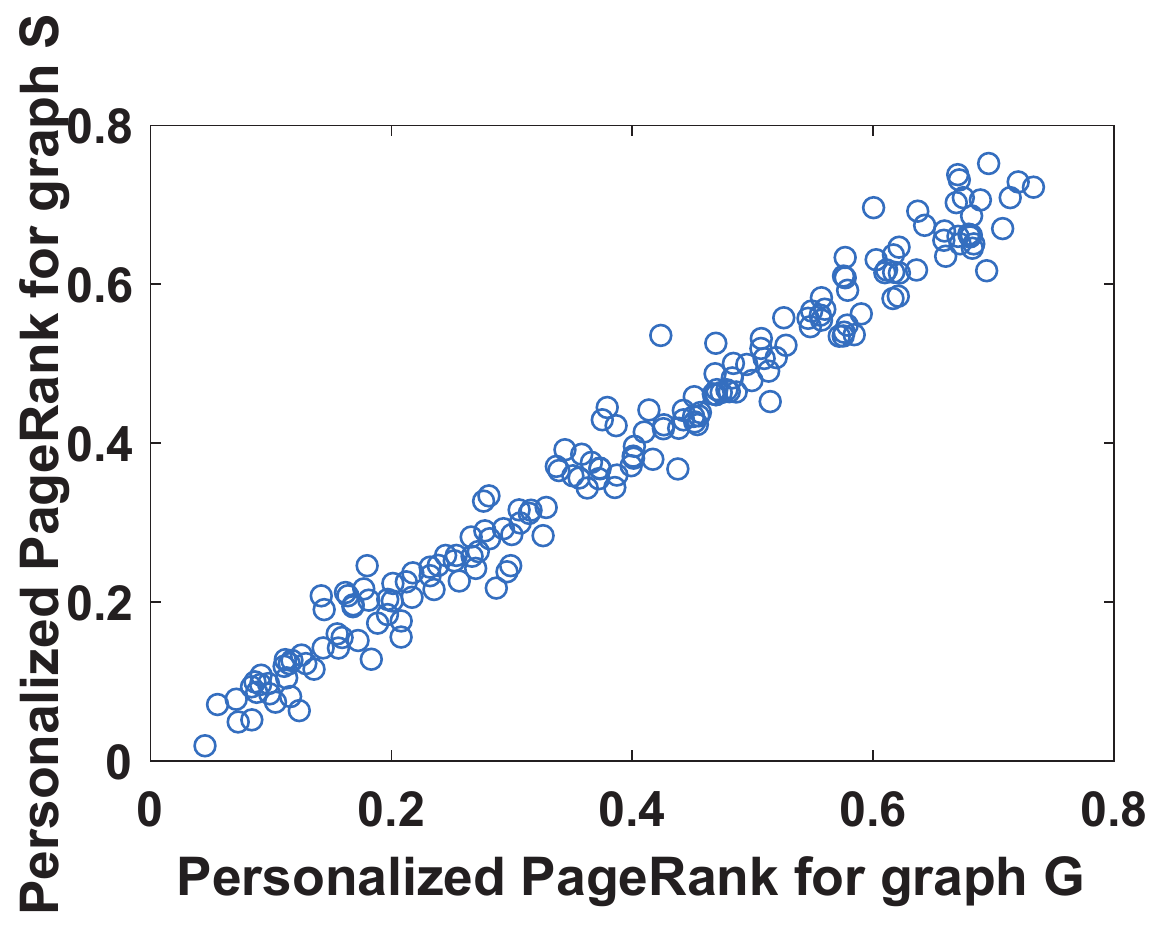} \caption{The correlation of the Personalized  PageRank between itself and its sparsifier for the 'gre$\_$115.mtx' graph  (left)  and  graph'gre$\_$185.mtx' (right) w/o smoothing. \protect\label{fig:personalized_pagerank}}
\vspace{0.15cm}
\end{figure}


Table \ref{table:fiedler1} shows the results of the directed Laplacian solver on different directed graphs. It reports relative errors between the exact solution and the solution calculated by the proposed solver with and without smoothing. It shows that errors can be dramatically reduced after smoothing, and our proposed solver can well approximate the true solution of $L_Gx=b$.
\begin{figure}[H]
\centering 
\includegraphics[scale=0.20193058]{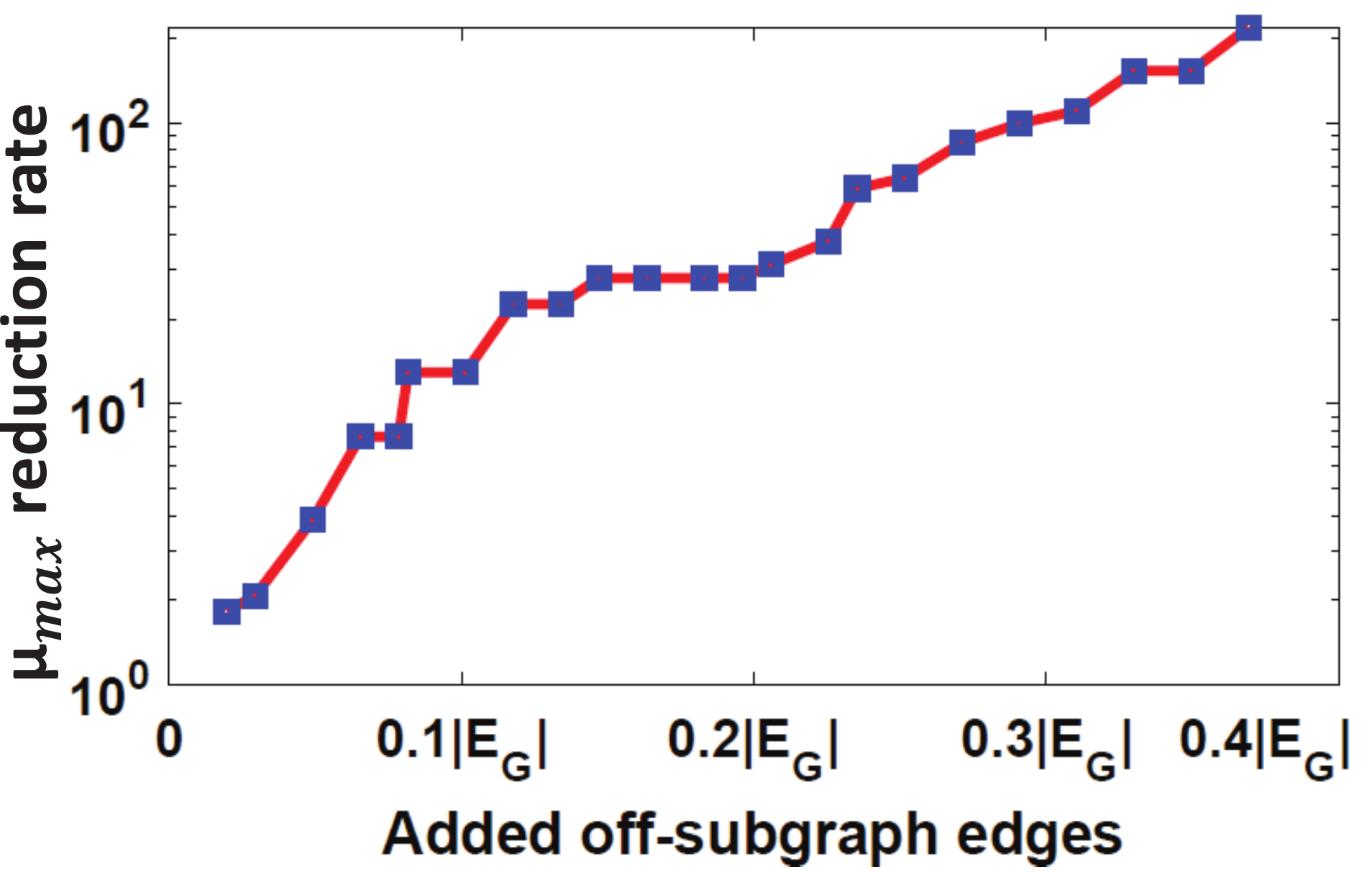}
%
\includegraphics[scale=0.20193058]{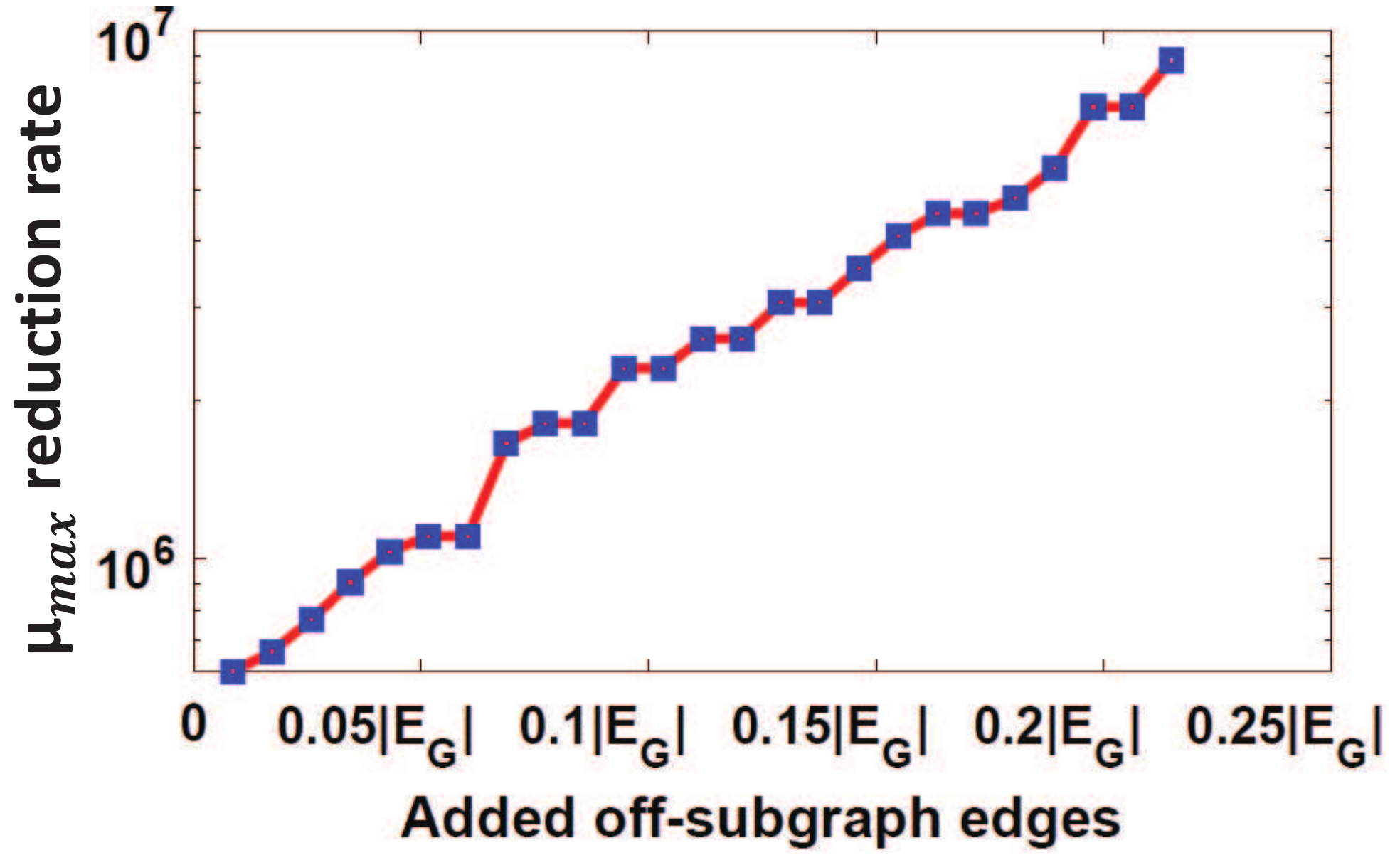} \caption{$\mu_{max}$ eigenvalue reduction rates for ``gre$\_$115" (left) and ``pesa" (right). \protect\label{fig:reduction_ratio}}
\end{figure}


Figure \ref{fig:personalized_pagerank} shows the personalized PageRank results on two graphs (``gre$\_$115" and ``gre$\_185$") and their sparsifiers. We can observe that personalized PageRank vectors match very well with the original ones.

Finally, we show the spectral graph partitioning results on the original directed graph Laplacian ${G_u}$ and its sparsifier ${S_u}$ in Figure \ref{fig:parition}. As observed, very similar partitioning results have been obtained, indicating  well preserved spectral properties within the spectrally-sparsified directed graph.

\begin{figure}[H]
\centering 
\includegraphics[scale=0.39304098]{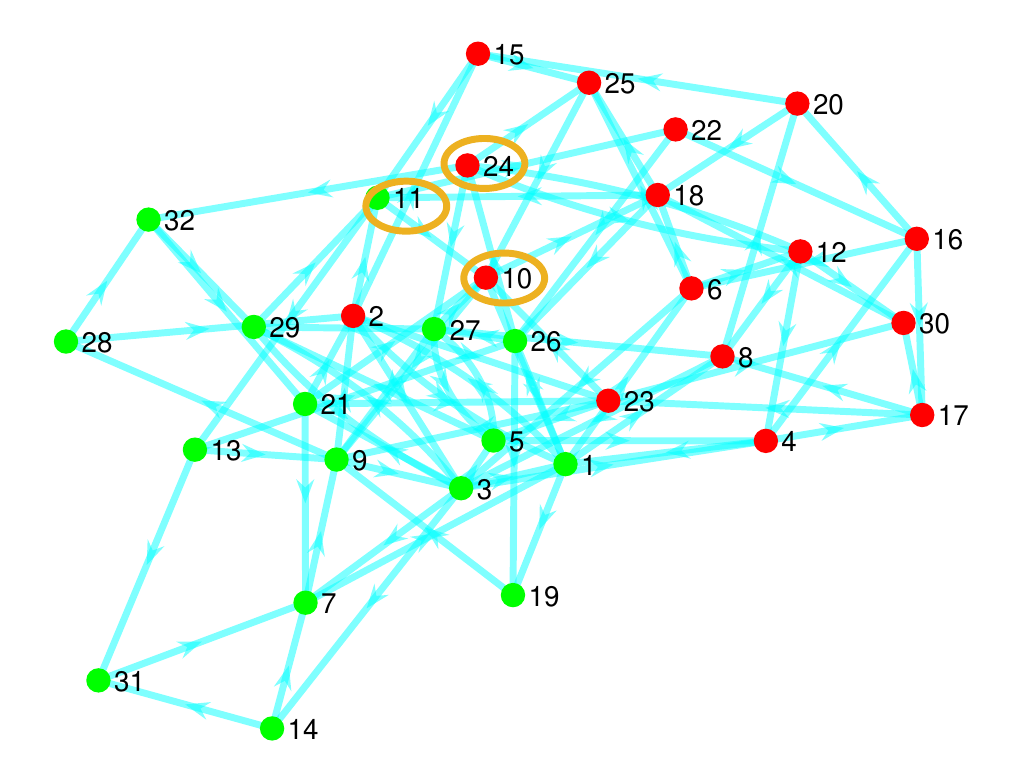}
\includegraphics[scale=0.39304098]{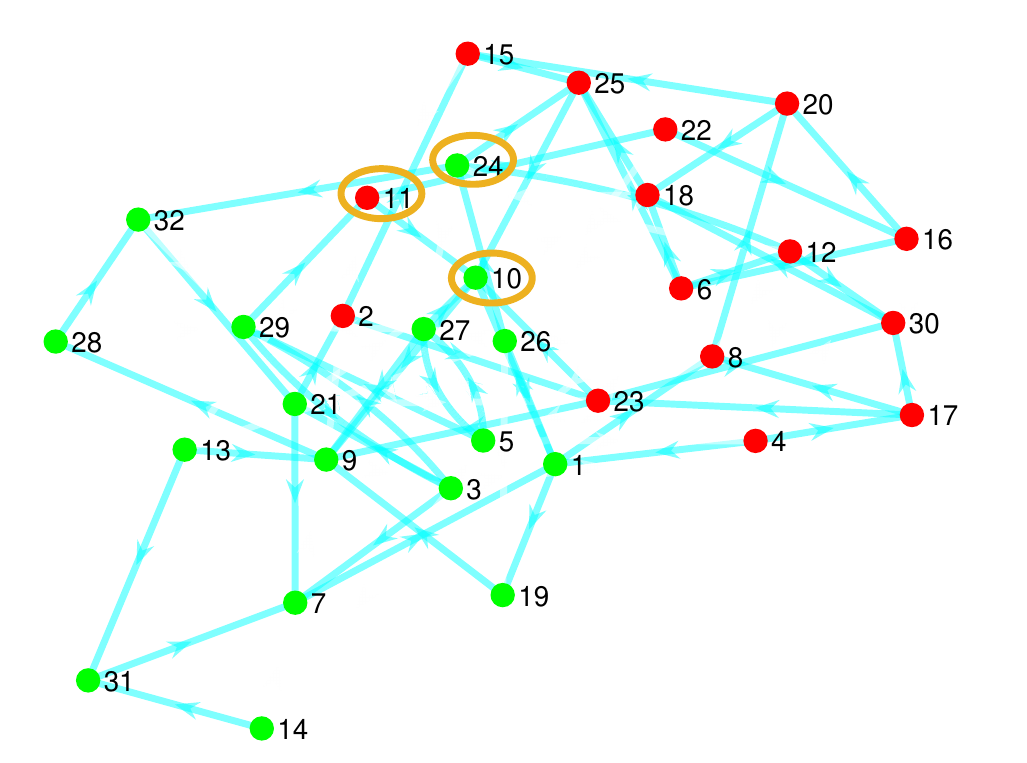} \caption{The partitioning results  between itself (left) and its sparsifier (right) for the 'ibm32.mtx' graph. \protect\label{fig:parition}}
\vspace{0.15cm}
\end{figure}





\vspace{-0.75cm}
\section{Conclusions}\label{conclusion}
This paper proves the existence of linear-sized spectral sparsifiers for general directed graphs, and proposes  a {{practically-efficient and unified spectral graph sparsification framework}}. Such a novel spectral sparsification approach allows sparsifying real-world, large-scale  directed and undirected graphs with guaranteed preservation of the original graph spectral properties. By exploiting a highly-scalable (nearly-linear complexity) spectral matrix perturbation analysis framework for constructing nearly-linear sized (directed) subgraphs,  it enables to well preserve the key eigenvalues and eigenvectors of the original (directed) graph Laplacians.  The proposed method has been validated using various kinds of directed graphs obtained from public domain sparse matrix collections, showing promising spectral sparsification and partitioning results for general directed graphs.
 \vspace{-0.17cm}
\bibliographystyle{abbrv}
{
\bibliography{dac18-spectralgraph,bigdata,graphRedu}  
}

\end{document}